\documentclass[reqno]{amsart}

\newtheorem{theorem}{Theorem}[section]

\theoremstyle{definition}

\usepackage[T1]{fontenc}
\usepackage[utf8]{inputenc}
\usepackage{lmodern}
\usepackage{textcomp}
\usepackage{microtype}
\usepackage{tikz}
\usepackage{mathtools, bm, nccmath}

\usepackage{graphicx}%
\usepackage{multirow}%
\usepackage{amsmath,amssymb,amsfonts}%
\usepackage{amsthm}%
\usepackage{mathrsfs}%
\usepackage[title]{appendix}%
\usepackage{xcolor}%
\usepackage{textcomp}%
\usepackage{manyfoot}%
\usepackage{booktabs}%
\usepackage{algorithm}%
\usepackage{algorithmicx}%
\usepackage{algpseudocode}%
\usepackage{listings}%
\usepackage[utf8]{inputenc}
\usepackage{url}
\usepackage{hyperref}
\usepackage{csquotes}
\usepackage[T1]{fontenc}
\usepackage{tikz}
\usepackage{graphicx}
\usepackage{subcaption}
\usepackage[english]{babel}
\usepackage{csquotes}
\usepackage[export]{adjustbox}
\usepackage{circuitikz}

\usepackage{tabularx}

\usepackage{graphicx}
\usepackage{setspace}

\usepackage{amsmath}
\usepackage{amsfonts}
\usepackage{amssymb}
\usepackage{amsthm}
\usepackage{xcolor}%
\usepackage{soul}
\usepackage{graphicx} 
\usepackage{lipsum}
\numberwithin{equation}{section}

\usepackage{thmtools}

\newtheorem{remark}[theorem]{Remark}

\numberwithin{equation}{section}

\makeatletter
\@namedef{subjclassname@2020}{\textup{2020} Mathematics Subject Classification}
\makeatother

\calclayout

\begin{document}
\title[The atmospheric Ekman spiral for piecewise-uniform eddy viscosity]
{The atmospheric Ekman spiral for piecewise-uniform eddy viscosity}

\author[E. Stefanescu]{Eduard Stefanescu}
\address{Institut für Mathematik, Universität Wien, Oskar-Morgenstern Platz 1, 1090 Wien, Austria}
\address{Institut f\"ur Analysis und Zahlentheorie, TU Graz, Steyrergasse 30, 8010 Graz, Austria}
\email{\href{mailto:eduard.stefanescu@tugraz.at}{eduard.stefanescu@tugraz.at}}


\begin{abstract}
We investigate the boundary-value problem of atmospheric Ekman flows with piecewise-uniform eddy viscosity. In addition we present a method for finding more general solutions by considering eddy viscosity as an arbitrary step-function. We discuss the existence and uniqueness of the solutions obtained through this method, providing detailed proofs for cases with one and two "jumps" in eddy viscosity. For scenarios with more "jumps," we establish results inductively.
Furthermore, we examine the angle between the bottom surface of the Ekman layer and geostrophic winds by extremizing variables such as the eddy viscosity and its point of change. These calculations reveal how the angle can differ from \(45^\circ\), demonstrating that the extreme values of \(0^\circ\) and \(90^\circ\) are achievable, indicating the potential range of the deflection angle.
\end{abstract}

\maketitle

\section{Introduction}\label{sec1}
In 1905 Ekman investigated the behaviour of wind-drift ocean currents, in the attempt to explain Nansen’s observation of the surface current deflection to the right in arctic regions. He wrote down an explicit solution for the velocity components dependent on height, assuming a constant eddy viscosity \cite{Ekman}.
The case consisting of continuous eddy viscosity has been covered by Constantin \cite{Consti4}, but even though it is analytically interesting, it poses the disadvantage of being mathematically rather complicated, by yielding an integral solution, which is rarely explicitly solvable.
Hence, assuming piecewise constant eddy viscosity gives an advantage of physically modelling more accurate as considering constant values and it also avoids rather complicated expressions. 
A drawback in limiting ourselves to piecewise eddy viscosity lies in losing differentiability, which mathematically can be overcome easily by a weak formulation. But physically speaking, by considering non-smoothness, we may lose accuracy to the continuous case; in particular when eddy viscosity fluctuates a lot. Nevertheless, we win accuracy to previous results, where interpreting solutions are mathematically similarly accessible.      
An explicit solution to the depth-dependent velocity components for variable eddy viscosity and for the deflection angle, i.e. the behaviour of the Ekman spiral under the sea surface, was constructed by D.G. Dritschel, N. Paldor, and A. Constantin \cite{Consti1} in 2020, by using similar methods as Ekman did. This was further improved by L. Roberti in \cite{Luigi}.
In this paper, related procedures as the above mentioned articles are used to describe similar phenomena in the atmosphere, where we also refine and extend the methods used above.

To give a better perspective into the motivation of this paper, we recollect ideas from \cite{Consti2} to give some background on atmospheric layers. The laminar sublayer with a thickness of a few millimetres is the lowest one and is not of relevance for the transfer of wind energy. The Prandtl layer is around $10$m, depending on the thermal stratification of the air. The upper layer is called the Ekman layer, where the airflow is driven by a balance of turbulent drag, Coriolis forces and pressure gradients. 
This layer ranges from the top of the Prandtl layer to heights typically between 100 and 1000 meters, depending on the geographic location and local atmospheric conditions. The Ekman layer covers approximately $90\%$ of the atmospheric boundary layer. Studies of this layer are broadly applicable.
For example the Ekman layer plays a crucial role in the formation and development of weather systems. It is where the frictional effects of the Earth's surface significantly influence the wind patterns, contributing to the development of cyclones and anticyclones, see e.g. \cite{holton2013introduction} and \cite{stull}. Furthermore the turbulent nature of the Ekman layer enhances the dispersion of pollutants. Understanding this layer helps in predicting how pollutants spread from industrial sites, urban areas, and other sources, which is essential for environmental management and public health, see e.g. \cite{seinfeld} and \cite{arya}.
\\
In the Introduction of \cite{Consti2}, there are different kinds of eddy viscosity with different behaviours listed. In the following there are suggestions that the eddy viscosity can increase, decrease, or even have other behaviours: \cite{Berger1998TheBV}, \cite{Grisogono1995AGE}, \cite{Tan}, \cite{2005BoLMe.115..399P}, \cite{Estoque},\cite{Madsen1977ARM}. Furthermore, numerical simulations were made in \cite{Grisogono23}, \cite{Marlatt} and \cite{Chai}. For more details we refer to \cite{holton2013introduction},\cite{marshall2008atmosphere}. 
\\

\section{Equations of motion}
The goal of this paper is to gain insight into the dynamics of 
mesoscale steady flows in the Ekman layer in non-equatorial regions of
the northern hemisphere. Decomposing the horizontal velocity into pressure-driven (geostrophic) and wind-driven (Ekman) components, we start with the equations 
\begin{eqnarray}\label{1-2}
(U-U_g)f &= \displaystyle\frac{{\rm d}}{{\rm d}Z}\,\Big( \nu(Z)\, \frac{{\rm d}V}{{\rm d}Z}\Big)\,,\\
(V-V_g)f &= - \displaystyle\frac{{\rm d}}{{\rm d}Z}\,\Big( \nu(Z)\, \frac{{\rm d}U}{{\rm d}Z}\Big)\,,
\end{eqnarray}
governing at leading order steady atmospheric Ekman flow at mid-latitudes in the $f$-plane approximation (see \cite{ekmantype}  and \cite{cushman2011introduction}). Here $Z$ 
measures altitude, $U$ and $V$ are the (horizontal) zonal and respectively meridional height-dependent
mean wind velocity components with corresponding geostrophic wind components $U_g$
and $V_g$, $\nu(Z)$ is the height-dependent eddy viscosity, and $f=2\Omega \sin \theta$ is the (constant) Coriolis parameter 
associated to the flow region located in the northern hemisphere near latitude $\theta \in \big(0,\frac{\pi}{2}\big)$, $\Omega \approx 7.29 \times 10^{-5}$ 
rad/s being the constant rate of rotation of the Earth around its polar axis. The boundary conditions for the system (1)-(2) are
\begin{eqnarray}
& U(0)=V(0)=0 \,,\\
& \lim\limits_{Z \to \infty} (U(Z),V(Z)) &= (U_g,V_g)\,,
\end{eqnarray}
expressing the fact that, due to friction,  the velocity vanishes at the flat bottom $Z=0$ of the Ekman layer, while the flow becomes geostrophic high up.  \\
We non-dimensionalize the system (1)-(4) by setting
$$
(U,V)=U^\ast\,(u,v)\,,\qquad (U_g,V_g)=U^\ast(u_g,v_g)\,,\qquad Z=H^\ast\,z\,,\,\qquad \nu(Z)=\nu^\ast \nu_0(z)\,,
$$
where $U^\ast$ is a typical horizontal speed (of the order of 10 m/s), 
$H^\ast$ is the typical height of the atmospheric boundary layer (of the order of a few hundred m) and $\nu^\ast$ is the average eddy viscosity. 
Introducing complex variables by setting $\psi=u+{\rm i}v$, $\psi_g=u_g+{\rm i}v_g$, and denoting 
$$K(z)=\frac{2 \nu^\ast}{f (H^\ast)^2}\,\nu_0(z)=\frac{2}{f (H^\ast)^2}\,\nu(Z)\,,$$
we obtain the non-dimensional version of (1)-(4),
\begin{equation}\label{5}
    (K\psi')' - 2{\rm i}(\psi-\psi_g)=0\ ,
\end{equation}
\begin{equation}\label{6}
     \psi(0)=0\ ,\qquad \lim\limits_{z \to \infty} \psi(z) = \psi_z\ ,
\end{equation}
where, for convenience, we write a prime for the derivative with respect to $z$.

\section{Exact solution for piecewise-constant eddy viscosity}\label{ch3}
Solutions to the problem above were discussed in $1905$ by Ekman, see \cite{Ekman}, but the function $K$ was considered to be constant. Motivated by \cite{Consti1} we show a method to solve the equations above with $K$ being a step function with finitely many steps, i.e. $K(z)=\sum_{n=0}^{N-1}l_n^2\mathbb{I}_{(a_n,a_{n+1})}(z)$, where $l_n\in\mathbb{R}^+=\{x\in\mathbb{R}|x> 0\}$, $l_j\ne l_{j+1}$, for $j=0,1,2,...,N-2$; $a_0=0$, $a_n\in\mathbb{R}$, $a_n<a_{n+1}$ for $n=1,2,...,N-1$ and $a_{N}=\infty$.\\ 
\subsection{Constructing the solution of the governing equations}\label{sec3.1}
For simplicity, we use a step function with only one jump, 
i.e. $K(z)=\mathbb{I}_{[0,h)}+l^2\mathbb{I}_{(h,\infty)}$. Inserting in \eqref{5} yields:
\begin{align}\label{8}
    \begin{split}
       \psi''-2i(\psi-\psi_g)&=0  \textnormal{ for } 0\le z< h, \\\
        l^2\psi''-2i(\psi-\psi_g)&=0  \textnormal{ for } h< z<\infty.
    \end{split}
\end{align}

Solving this well known non-homogeneous linear differential equation of second order yields the solution: 
\begin{align}\label{33}
    \begin{split}
       \psi(z)&=Ae^{(1+i)z}+Be^{-(1+i)z}+\psi_g\textnormal{ for }0\le z< h, \\\
       \psi(z)&=Ce^{\frac{(1+i)z}{l}}+De^{-\frac{(1+i)z}{l}}+\psi_g\textnormal{ for }h< z< \infty
    \end{split}
\end{align}
We need four conditions to get explicit values for the constants $A,B,C,D$.
The right hand side of condition \eqref{6}, expressed as $\lim\limits_{z \rightarrow \infty}\psi(z)=\psi_g$, indicates that $\psi$ must be bounded. Consequently, the constant $C$ must be zero, since $\exp\left((1+i)z/l\right)$ is unbounded. Additionally, the left hand side of condition \eqref{6}, which states $\psi(0)=0$, leads to the equation $A+B+\psi_g=0$. 
Furthermore, we assume the function $\psi$ to be continuous, which means that $\psi(h^-)=\psi(h^+)$. Inserting this in \eqref{33} generates the equation $$Ae^{(1+i)h}+Be^{-(1+i)h}=De^{-\frac{(1+i)h}{l}}.$$ 
Integrating the above equations against testfunctions and using integration by parts lead to $$-\psi_-'(h)+l^2\psi_+'(h)=0.$$ 
This in particular yields
\begin{equation}\label{asbefore}
    Ae^{(1+i)h}-Be^{-(1+i)h}=-lDe^{-\frac{(1+i)h}{l}}.
\end{equation}
The above conditions are a system of linear equations. Solving this yields:
\begin{equation}\label{38}
\begin{split}
     A=\frac{\psi_g(l-1)}{1+e^{(2+2i)h}-l+e^{(2+2i)h}l},\ B=-\frac{\psi_ge^{(2+2i)h}(l+1)}{1+e^{(2+2i)h}-l+e^{(2+2i)h}l},\\\
   C=0,\ D=-\frac{2\psi_ge^{(1+i)h+\frac{(1+i)h}{l}}}{1+e^{(2+2i)h}-l+e^{(2+2i)h}l}.    
\end{split}  
\end{equation}
Thus, Equation \eqref{8} is fully solved.

\subsubsection{Solving method of the equations of motion for arbitrary piecewise constant eddy viscosity}
Let us consider the general case, where $K$ is a step function with finitely many steps, defined as in the beginning of Chapter \ref{ch3}.
We calculate all the ODEs in each height section the same way as before and get a solution with $2N$ unknown constants: 
\begin{align}
\begin{split}\label{318}
    \psi(z)&=A_{0,0}e^{\frac{(1+i)z}{l_0}}+A_{0,1}e^{-\frac{(1+i)z}{l_0}}+\psi_g\textnormal{ for } a_0=0\le z< a_1,\\\
    \psi(z)&=A_{1,0}e^{\frac{(1+i)z}{l_1}}+A_{1,1}e^{-\frac{(1+i)z}{l_1}}+\psi_g \textnormal{ for } a_1\le z<a_2,\\\
    &\vdots\\\
    \psi(z)&=A_{N-1,0}e^{\frac{(1+i)z}{l_{N-1}}}+A_{N-1,1}e^{-\frac{(1+i)z}{l_{N-1}}}+\psi_g   \textnormal{ for } a_{N-1}\le z< \infty.
\end{split}
\end{align}
So we need to write down $2N$ linear equations including these $2N$ unknown constants, that are linear independent from each other. 
We get two equations from the two boundary conditions \eqref{6}:
\begin{equation}\label{319}
    A_{0,0}+A_{0,1}+\psi_g=0,
\end{equation}
\begin{equation}\label{320}
    A_{N-1,0}=0.
\end{equation}
By assuming continuity of $\psi$ again, we get for each of the $(N-1)$ jumps a linear equation:
\begin{equation}\label{321}
\begin{split}
    A_{0,0}e^{\frac{(1+i)a_1}{l_0}}+A_{0,1}e^{-\frac{(1+i)a_1}{l_0}}=A_{1,0}e^{\frac{(1+i)a_1}{l_1}}+A_{1,1}e^{-\frac{(1+i)a_1}{l_1}},\\\
    A_{1,0}e^{\frac{(1+i)a_2}{l_1}}+A_{1,1}e^{-\frac{(1+i)a_2}{l_1}}=A_{2,0}e^{\frac{(1+i)a_2}{l_2}}+A_{2,1}e^{-\frac{(1+i)a_2}{l_2}},\\\
    \vdots\\\
    A_{N-2,0}e^{\frac{(1+i)a_{N-1}}{l_{N-2}}}+A_{N-2,1}e^{-\frac{(1+i)a_{N-1}}{l_{N-2}}}=A_{N-1,1}e^{\frac{(1+i)a_{N-1}}{l_{N-1}}}.
\end{split}  
\end{equation}
Finally by the same method as in \eqref{asbefore}, we get the remaining $N-1$ equations: 
\begin{equation}\label{323}
\begin{split}
    l_0A_{0,0}e^{\frac{(1+i)a_1}{l_0}}-l_0A_{0,1}e^{-\frac{(1+i)a_1}{l_0}}=l_1A_{1,0}e^{\frac{(1+i)a_1}{l_1}}-l_1A_{1,1}e^{-\frac{(1+i)a_1}{l_1}},\\\
    l_1A_{1,0}e^{\frac{(1+i)a_2}{l_1}}-l_1A_{1,1}e^{-\frac{(1+i)a_2}{l_1}}=l_2A_{2,0}e^{\frac{(1+i)a_2}{l_2}}-l_2A_{2,1}e^{-\frac{(1+i)a_2}{l_2}},\\\
    \vdots\\\
    l_{N-2}A_{N-2,0}e^{\frac{(1+i)a_{N-1}}{l_{N-2}}}-l_{N-2}A_{N-2,1}e^{-\frac{(1+i)a_{N-1}}{l_{N-2}}}=-l_{N-1}A_{N-1,1}e^{\frac{(1+i)a_{N-1}}{l_{N-1}}}.
    \end{split}
\end{equation}
The corresponding system of equations can be written in the following form:
\begin{equation}\label{324}
\mathcal{M}\mathbf{A}=\mathbf{b}, \textnormal{with}
\end{equation}
$$\mathbf{A}=\left(A_{0,0},A_{0,1},A_{1,0},A_{1,1},...,A_{N-2,0},A_{N-2,1},A_{N-1,1},A_{N-1,0}\right)^T,$$
$$\mathbf{b}=\left(-\psi_g,0,...,0\right)^T,$$
$$\mathcal{M}=$$
\tiny
$$\left[\begin{array}{rrrrrrrrrrr} 
1 & 1 & 0 & 0 & 0 & 0 & \cdots & \cdots  & \cdots&  0 & 0\\ 
\alpha_0 & \hat{\alpha}_0 & -\beta_1 & -\hat{\beta}_1 & 0 & 0 & \ddots & \ddots & \ddots & 0 & 0\\ 
0 & 0 & \alpha_1 & \hat{\alpha}_1 & -\beta_2 & -\hat{\beta}_2 & \ddots & \ddots  & \ddots & 0 & 0\\ 
\vdots & \vdots & \ddots & \ddots & \ddots &  \ddots & \ddots & \ddots & \ddots & \vdots & \vdots\\
0 & 0 & 0 & 0 & \cdots & \cdots & \cdots & \alpha_{N-2} & \hat{\alpha}_{N-2} & -\hat{\beta}_{N-1} & 0\\ 
l_0\alpha_0 & -l_0\hat{\alpha}_0 & -l_1\beta_1 & l_1\hat{\beta}_1 & 0 & 0 & \cdots & \cdots & \cdots &  0 & 0\\ 
0 & 0 & l_1\alpha_1 & -l_1\hat{\alpha}_1 & -l_2\beta_2 & l_2\hat{\beta}_2 & \ddots &\ddots & \ddots & 0 & 0\\ 
\vdots & \vdots & \ddots & \ddots & \ddots & \ddots & \ddots & \ddots & \ddots &\vdots & \vdots\\ 
0 & 0 & 0 & 0 & \cdots & \cdots & \cdots & l_{N-2}\alpha_{N-2} & -l_{N-2}\hat{\alpha}_{N-2} &  l_{N-1}\hat{\beta}_{N-1} & 0\\ 
0 & 0 & 0 & 0 & \cdots & \cdots & \cdots & 0 & 0 & 0 & 1\\ 
\end{array}\right],$$
\normalsize
where $\alpha_n:=e^{\frac{(1+i)a_{n+1}}{l_n}}$, $\hat{\alpha}_n:=e^{-\frac{(1+i)a_{n+1}}{l_n}}$ $\beta_n:=e^{\frac{(1+i)a_{n}}{l_n}}$, $\hat{\beta}_n:=e^{-\frac{(1+i)a_{n}}{l_n}}$.
\subsection{Linear independence}
When trying to solve a system of $2N$ equations with $2N$ unknowns we need to find out if there exists a solution and if it is unique, which is equivalent to the associated square matrix being invertible.
\subsubsection{$4\times4$-Matrix}\label{4x4}
First we only consider one jump at the point $a_1$, which means that $K(z)=l_0^2\mathbb{I}_{[0,a_1)}+l_1^2\mathbb{I}_{(a_1,\infty)}$. 
This is the case we discussed in Section \ref{sec3.1}, so when trying to calculate the solutions explicitly, we can use the following method:
By \eqref{319},\eqref{320} we have
\begin{equation}\label{3311}
    A_{0,0}+A_{0,1}+\psi_g=0,\quad A_{1,0}=0.
\end{equation}
We rewrite \eqref{321},\eqref{323}:
\begin{equation}\label{312}
\begin{bmatrix}
    \alpha_0 & \hat{\alpha}_0\\
    l_0\alpha_0 & -l_0\hat{\alpha}_0\\ 
\end{bmatrix}
\begin{pmatrix}
\vphantom{1} A_{0,0}\\
\vphantom{e^{-(1+i)h}} A_{0,1}\\
\end{pmatrix} = 
\begin{bmatrix}
    \beta_1 & \hat{\beta}_1\\
    l_1\beta_1 & -l_1\hat{\beta}_1\\ 
\end{bmatrix}
\begin{pmatrix}
\vphantom{1} 0\\
\vphantom{e^{-(1+i)h}} A_{1,1}\\
\end{pmatrix},
\end{equation}
where $\alpha_0$, $\beta_1$ are defined as in the previous section.
Inverting the first matrix on the left hand side is possible, since its determinant is $-2l_0$. Hence we get:
\begin{equation}\label{3314}
\begin{pmatrix}
\vphantom{1} A_{0,0}\\
\vphantom{e^{-(1+i)h}} A_{0,1}\\
\end{pmatrix} = 
\begin{bmatrix}
    \alpha_0 & \hat{\alpha}_0\\
    l_0\alpha_0 & -l_0\hat{\alpha}_0\\ 
\end{bmatrix}^{-1}
\begin{bmatrix}
    \beta_1 & \hat{\beta}_1\\
    l_1\beta_1 & -l_1\hat{\beta}_1\\ 
\end{bmatrix}
\begin{pmatrix}
\vphantom{1} 0\\
\vphantom{e^{-(1+i)h}} A_{1,1}\\
\end{pmatrix},
\end{equation}
This means that $A_{0,0}$ and $A_{0,1}$ are just multiples of $A_{1,1}$, i.e. there exists $\lambda_{0,0},\lambda_{0,1}$ such that $A_{0,0}=\lambda_{0,0}A_{1,1}$ and $A_{0,1}=\lambda_{0,1}A_{1,1}$. Inserting this into \eqref{3311} yields:
\begin{equation}
    A_{1,1}=\frac{-\psi_g}{\lambda_{0,0}+\lambda_{0,1}}.
\end{equation}
Since $A_{0,0}$ and $A_{0,1}$ are dependent on $A_{1,1}$ the last thing remaining is to check if $\lambda_{0,0}+\lambda_{0,1}\ne 0$ to get a unique solution to all coefficients.
\begin{theorem} Let $\lambda_{0,0}$ and $\lambda_{0,1}$ be defined as above. Then \label{thm1}
    $\lambda_{0,0}+\lambda_{0,1}\ne 0$.
\end{theorem}
\begin{proof}
    By \eqref{3314} we get $\lambda_{0,0},\lambda_{0,1}$ explicitly:
$$\lambda_{0,0}=\frac{e^{\frac{-(1+i)a_1(l_1+l_0)}{l_1l_0}}(-l_1+l_0)}{2l_0},\quad\lambda_{0,1}=\frac{e^{\frac{(1+i)a_1(l_1-l_0)}{l_1l_0}}(l_1+l_0)}{2l_0},$$
Now, we assume $\lambda_{0,0}+\lambda_{0,1}= 0$, which is equivalent to
\begin{align}
    \begin{split}
        (-1+e^{\frac{(2+2i)a_1}{l_0}})l_1+(1+e^{\frac{(2+2i)a_1}{l_0}})l_0&=0\\\
        \iff \frac{l_0-l_1}{l_0+l_1}&=-e^{\frac{(2+2i)a_1}{l_0}}.
    \end{split}
\end{align}
But we have $|\frac{l_0-l_1}{l_0+l_1}|<1$, by the positivity of $l_0,l_1$ and the inverse triangle inequality and $|-e^{\frac{(2+2i)a_1}{l_0}}|=|e^{2\frac{a_1}{l_0}}|>1$, by the positivity of $l_0$ and $a_1$. This is a contradiction to the assumption.
\end{proof}
Thus, the solution to all coefficients exists and is unique.
\subsubsection{$6\times 6$-Matrix}\label{6x6}
Similar as before we rewrite \eqref{321}, \eqref{323}:
\begin{equation}\label{33141}
\begin{pmatrix}
\vphantom{1} A_{0,0}\\
\vphantom{e^{-(1+i)h}} A_{0,1}\\
\end{pmatrix} = 
\begin{bmatrix}
    \alpha_0 & \hat{\alpha}_0\\
    l_0\alpha_0 & -l_0\hat{\alpha}_0\\ 
\end{bmatrix}^{-1}
\begin{bmatrix}
    \beta_1 & \hat{\beta}_1\\
    l_1\beta_1 & -l_1\hat{\beta}_1\\ 
\end{bmatrix}
\begin{pmatrix}
\vphantom{1} A_{1,0}\\
\vphantom{e^{-(1+i)h}} A_{1,1}\\
\end{pmatrix},
\end{equation}
\begin{equation}\label{33142}
\begin{pmatrix}
\vphantom{1} A_{1,0}\\
\vphantom{e^{-(1+i)h}} A_{1,1}\\
\end{pmatrix} = 
\begin{bmatrix}
    \alpha_1 & \hat{\alpha}_1\\
    l_1\alpha_1 & -l_1\hat{\alpha}_1\\ 
\end{bmatrix}^{-1}
\begin{bmatrix}
    \beta_2 & \hat{\beta}_2\\
    l_2\beta_2 & -l_2\hat{\beta}_2\\ 
\end{bmatrix}
\begin{pmatrix}
\vphantom{1} 0\\
\vphantom{e^{-(1+i)h}} A_{2,1}\\
\end{pmatrix}.
\end{equation}
Consequently $A_{i,j};i,j=0,1$ are multiples of $A_{2,1}$, i.e. there exists $\lambda_{i,j}$ such that $A_{i,j}=\lambda_{i,j}A_{2,1};i,j=0,1$.
In particular we can write:
\begin{equation}\label{3316}
\begin{pmatrix}
\vphantom{1} A_{0,0}\\
\vphantom{e^{-(1+i)h}} A_{0,1}\\
\end{pmatrix} = 
\begin{bmatrix}
    \alpha_0 & \hat{\alpha}_0\\
    l_0\alpha_0 & -l_0\hat{\alpha}_0\\ 
\end{bmatrix}^{-1}
\begin{bmatrix}
    \beta_1 & \hat{\beta}_1\\
    l_1\beta_1 & -l_1\hat{\beta}_1\\ 
\end{bmatrix}
\begin{bmatrix}
    \alpha_1 & \hat{\alpha}_1\\
    l_1\alpha_1 & -l_1\hat{\alpha}_1\\ 
\end{bmatrix}^{-1}
\begin{bmatrix}
    \beta_2 & \hat{\beta}_2\\
    l_2\beta_2 & -l_2\hat{\beta}_2\\ 
\end{bmatrix}
\begin{pmatrix}
\vphantom{1} 0\\
\vphantom{e^{-(1+i)h}} A_{2,1}\\
\end{pmatrix}.
\end{equation}
The above is well defined, since the matrices that were inverted have determinant $-2l_i$, $i=0,1$.
Again by the ''zero boundary condition'' and the same calculation as above we have: 
\begin{equation}
    A_{2,1}=\frac{-\psi_g}{\lambda_{0,0}+\lambda_{0,1}}.
\end{equation}
Hence, to show existence and uniqueness of all coefficients, we only need to check that $\lambda_{0,0}+\lambda_{0,1}\ne 0$.
\begin{theorem}
Let $\lambda_{0,0}$ and $\lambda_{0,1}$ be defined as above. Then  $\lambda_{0,0}+\lambda_{0,1}\ne 0$.   
\end{theorem}
\begin{proof}
By \eqref{3316} we can calculate $\lambda_{0,0},\lambda_{0,1}$ explicitly and write as before:
$$\lambda_{0,0}+\lambda_{0,1}=\frac{e^{-\frac{(1+i)a_2(l_1+l_2)}{l_0l_1}}}{4l_0l_1}[x(c_1\eta_2-c_2\theta_1)+y(-c_1\eta_1+c_2\theta_2)],$$
with $$x=e^{\frac{-(1+i)a_1(l_1+l_0)}{l_0l_1}}, y=e^{\frac{(1+i)a_1(-l_1+l_0)}{l_0l_1}},$$
$$c_1=e^{\frac{(2+2i)a_1}{l_1}},c_2=e^{\frac{(2+2i)a_2}{l_1}},$$
$$\eta_1=(l_1-l_0)(l_1-l_2),\eta_2=(l_1+l_0)(l_1-l_2),$$
$$\theta_1=(l_1-l_0)(l_1+l_2),\theta_2=(l_1+l_0)(l_1+l_2).$$
Since the first factor is nonzero we only need to check $x(c_1\eta_2-c_2\theta_1)+y(c_1\eta_1-c_2\theta_2)\ne 0$.
\textbf{Claim 1:} $|c_2\theta_2-c_1\eta_1|>|c_2\theta_1-c_1\eta_2|$.\\
An equivalent statement that is true:
$$|c_2\theta_2-c_1\eta_1|^2>|c_2\theta_1-c_1\eta_2|^2$$
By simple mathematical identities and by the fact that $\theta_2\eta_1=\theta_1\eta_2$, this is equivalent to
$$ 2e^{\frac{4a_2}{l_1}}(2l_1^3l_0+4l_1^2l_0l_2+2l_0l_1l_2)>2e^{\frac{4a_1}{l_1}}(2l_1^3l_0-4l_1^2l_0l_2+2l_0l_1l_2)$$
We have that $e^{\frac{4a_2}{l_1}}>e^{\frac{4a_1}{l_1}}$ by the fact that $a_2>a_1$. 
Furthermore $(2l_1^3l_0+4l_1^2l_0l_2+2l_0l_1l_2)>(2l_1^3l_0-4l_1^2l_0l_2+2l_0l_1l_2)$ by the fact that all $l_i>0, i=0,1,2$.
So the last expression is true, which means that the first expression of this claim is true and claim 1 is proven.\\
\textbf{Claim 2:} $\lambda_{0,0}+\lambda_{0,1}\ne 0$.\\
By the fact that $|\frac{y}{x}|>1$, we have that
$$|c_2\theta_2-c_1\eta_1|>|c_2\theta_1-c_1\eta_2|,$$
implies
$$ y(-c_1\eta_1+c_2\theta_2)\ne -x(c_1\eta_2-c_2\theta_1).$$
Hence $\lambda_{0,0}+\lambda_{0,1}\ne 0.$

\end{proof}
\subsubsection{$2n\times2n$-Matrix}
First we define $$A_{k+1}:=\begin{bmatrix}
\alpha_k & \hat{\alpha}_k\\
l_k\alpha_k & -l_k\hat{\alpha}_k
\end{bmatrix}, B_k=\begin{bmatrix}
\beta_k & \hat{\beta}_k\\
l_k\beta_k & -l_k\hat{\beta}_k
\end{bmatrix},$$   
$$A_1^{-1}B_1A_2^{-1}B_2...A_k^{-1}B_k:=\begin{bmatrix}
w_k & x_k\\
y_k & z_k
\end{bmatrix}.$$
The above is well defined, since $\det(A_i)=-2l_{i-1},i=1,...,n$.
Similar as before we can write
\begin{equation}\label{234}
\begin{pmatrix}
\vphantom{1} A_{1,0}\\
\vphantom{e^{-(1+i)h}} A_{1,1}\\
\end{pmatrix} = 
A_1^{-1}B_1A_2^{-1}B_2...A_n^{-1}B_n
\begin{pmatrix}
\vphantom{1} 0\\
\vphantom{e^{-(1+i)h}} A_{n,1}\\
\end{pmatrix},
\end{equation}
and $A_{i,j};i=0,...n-1;j=0,1$ are multiples of $A_{n,1}$, i.e. there exists $\lambda_{i,j}$ such that $A_{i,j}=\lambda_{i,j}A_{2,1}; i=0,...,n-1; j=0,1$.
Again, if $\lambda_{0,0}+\lambda_{0,1}\ne 0$ we get an unique solution for $A_{n,1}$, hence we get an unique solution for all coefficients.\\
We observe that:
\begin{equation}\label{2N}
 A_k^{-1}B_k=\begin{bmatrix}
-\hat{\alpha}_{k-1}\beta_k(l_k+l_{k-1}) & \hat{\alpha}_{k-1}\hat{\beta}_k(l_k-l_{k-1})\\
\alpha_{k-1}\beta_k(l_k-l_{k-1}) & -\alpha_{k-1}\hat{\beta}_k(l_k+l_{k-1})
\end{bmatrix}.   
\end{equation}
The idea of the proof in the general case is to assume the viscosities $l_i,i=0,...,n$ to be close to each other, which makes sense in physics, when assuming numerous different values for them. When doing so the matrices $A_k^{-1}B_k$, for $k=1,...,n$ becomes ''nearly'' diagonal. Using the structure of \eqref{234}, $\lambda_{0,0}$ can be made arbitrary small. Consequently $\lambda_{0,0}+\lambda_{0,1}\ne 0$ holds.
\begin{theorem} Let all variables be defined as above. Then 
there exists an $\varepsilon_k>0, k=0,1,...,n-1$, such that $\frac{|l_{k+1}-l_k|}{|l_{k+1}+l_k|}\le\varepsilon_k$ implies:
$$\lambda_{0,0}+\lambda_{0,1}\ne 0.$$
In particular this means that for all $k$, if we choose $l_{k+1}$ close enough to $l_{k}$ such that for a given $\varepsilon_k$ the above holds, then there exists an unique solution to all unknown constants in \eqref{318}.
\end{theorem}
\begin{proof}
We prove inductively, that $x_n+z_n\ne 0$, $\forall n\in\mathbb{N}$, which is all we need, since $x_n=\lambda_{0,0}$ and $z_n=\lambda_{0,1}$. \\
\textbf{Base case:} This is done in Theorem \ref{thm1} and holds without further assumptions.\\
\textbf{Assumption:} We can now assume that $x_k+ z_k\ne 0$. This implies there exists an $\tilde{\varepsilon_k}$ such that $|x_k+z_k|>\tilde{\varepsilon}_k$. Choose $\varepsilon_k$ such that $\tilde{\varepsilon}_k=\frac{|w_k-y_k|}{\alpha_k^2}\varepsilon_k$.\\
\textbf{Induction Step:} By definition we have:
$$\begin{bmatrix}
w_k & x_k\\
y_k & z_k
\end{bmatrix}A_{k+1}^{-1}B_{k+1}=\begin{bmatrix}
w_{k+1} & x_{k+1}\\
y_{k+1} & z_{k+1}
\end{bmatrix}.$$
Calculating the above with the help of \eqref{2N} we get:
\begin{equation}
    \begin{split}
       x_{k+1}=w_k\hat{\alpha}_{k}\hat{\beta}_{k+1}(l_{k+1}-l_{k})-x_k\alpha_k\hat{\beta}_{k+1}(l_{k+1}+l_{k}),\\
z_{k+1}=y_k\hat{\alpha}_{k}\hat{\beta}_{k+1}(l_{k+1}-l_{k})-z_k\alpha_k\hat{\beta}_{k+1}(l_{k+1}+l_{k}). 
    \end{split}
\end{equation}
This yields:
\begin{align}
    \begin{split}
        x_{k+1}+z_{k+1}&=0\\\
        \iff(l_{k+1}-l_k) (w_k+y_k)\hat{\alpha}_k\hat{\beta}_{k+1}-(x_k+z_k) (l_{k+1}+l_k) \alpha_k\hat{\beta}_{k+1}&=0\\\
        \iff\frac{w_k-y_k}{\alpha_k^2}\frac{l_{k+1}-l_k}{l_{k+1}+l_k}&=(x_k+z_k).
    \end{split}
\end{align}
If we now choose $l_{k+1}$ close enough to $l_k$ such that the assumption holds, then we have:
$$\tilde{\varepsilon}_k<|x_k+z_k|=\frac{|w_k-y_k|}{|\alpha_k^2|}\frac{|l_{k+1}-l_k|}{|l_{k+1}+l_k|}\le\frac{|w_k-y_k|}{|\alpha_k^2|}\varepsilon_k=\tilde{\varepsilon}_k$$
This is a contradiction. Hence $|x_{k+1}+z_{k+1}|>0,$ which implies $x_{k+1}+z_{k+1}\ne 0.$

\end{proof}
\begin{remark}
The case of a piecewise-constant eddy vorticity with one jump discontinuity for the oceanic near-surface Ekman layer was studied recently in \cite{Consti1}. The considerations in the present paper are devoted to atmospheric piecewise-constant eddy vorticities with $N \in \{1,2,3\}$ arbitrary jump discontinuities and $N >3$ small jumps in the eddy viscosity coefficient. We expect that with increasing $N$ there should be a convergence toward the solution with a continuous eddy viscosity coefficient. However, for a general continuous eddy viscosity the study of the behaviour of the corresponding solution is quite challenging (see the discussion in \cite{Consti4} for the case of the ocean), so that an approximation with piecewise-constant functions having finitely many jumps is helpful. Gradually increasing number of jumps should, in principle, give an increasingly closer approximation to the solution corresponding to a continuous eddy viscosity. We believe that these are important directions for further studies.
\end{remark}

\section{The Ekman spiral for piecewise-uniform eddy viscosity}
We define the angle $\gamma$ between the wind at any height and that of the geostrophic vector in the same manner as in \cite{Consti2}:
\begin{equation}\label{311}
    \tan(\gamma(z))=\frac{\Im\left(\frac{\psi(z)}{\psi_g}\right)}{\Re\left(\frac{\psi(z)}{\psi_g}\right)}.
\end{equation}

\subsection{The surface deflection angle}
In particular we calculate the angle between the surface wind and that of the geostrophic vector. For brevity, we refer to this angle as the deflection angle. After inserting $\psi(0)$ into Equation \eqref{311} we see from the initial condition \eqref{6} that L'Hospital's rule is needed. This means
\begin{equation}\label{3115}
    \lim\limits_{z \rightarrow 0}\frac{\Im\left(\frac{\psi(z)}{\psi_g}\right)}{\Re\left(\frac{\psi(z)}{\psi_g}\right)}=\lim\limits_{z \rightarrow 0}\frac{\Im\left(\frac{\psi(z)}{\psi_g}\right)'}{\Re\left(\frac{\psi(z)}{\psi_g}\right)'}.
\end{equation}
We can immediately check that $\psi'(z)=(1+i)(A+B)$ for $0\le z< h$. Inserting this into Equation \eqref{3115} and using that $e^{i\theta}=\cos(\theta)+i\sin(\theta)$ for $\theta\in\mathbb{R}$, we get
\begin{equation}\label{312}
    \tan(\gamma_0)=\frac{\alpha^2-\beta^2+2\alpha\beta\sin(2h)}{\alpha^2-\beta^2-2\alpha\beta\sin(2h)},
\end{equation}
where $\alpha:=(1+i)e^h$ and $\beta:=(1-i)e^{-h}$.

\section{Interpretations and Graphs}

Considering \cite{Consti2} directly or taking values from \cite{holton2013introduction}, and combining them with the modelling in \cite{ekmantype}, we get values $l$ of order $1$ to $0.01$, which for brevity we call eddy viscosity. 
Increasing jump-points $h$ will affect the convergence-rate of the angle towards balance, as it is clearly visible from the derived equations and will be discussed in Chapter \ref{anaprop}.
\subsection{Two examples for non-$45^{\circ}$-degree deflection angles}
We showcase two examples, where we let lower layer (up to the jump point $h$) eddy viscosity be $1$ as before and choose upper layer eddy viscosity $l=5$ and $l=0.08$, representing high and low end values, see \cite{Consti2}. Furthermore, we selected $h=1.1$ and $h=0.35$ to demonstrate the convergence rates.
In Figure \ref{fig:sub1}, where low eddy-viscosity is chosen, we observe a deflection angle over $45^{\circ}$, whereas a deflection angle under $45^{\circ}$ occurs in Figure \ref{fig:sub2}, taking larger eddy viscosity for $l$. Moreover, an obvious difference in the convergence rate towards rotational balance of the angle can be observed, which is caused by the choice of the jump point $h$.

\begin{figure}[ht!]
    \centering
    \begin{subfigure}[b]{0.47\textwidth} 
        \centering
        \includegraphics[width=\linewidth]{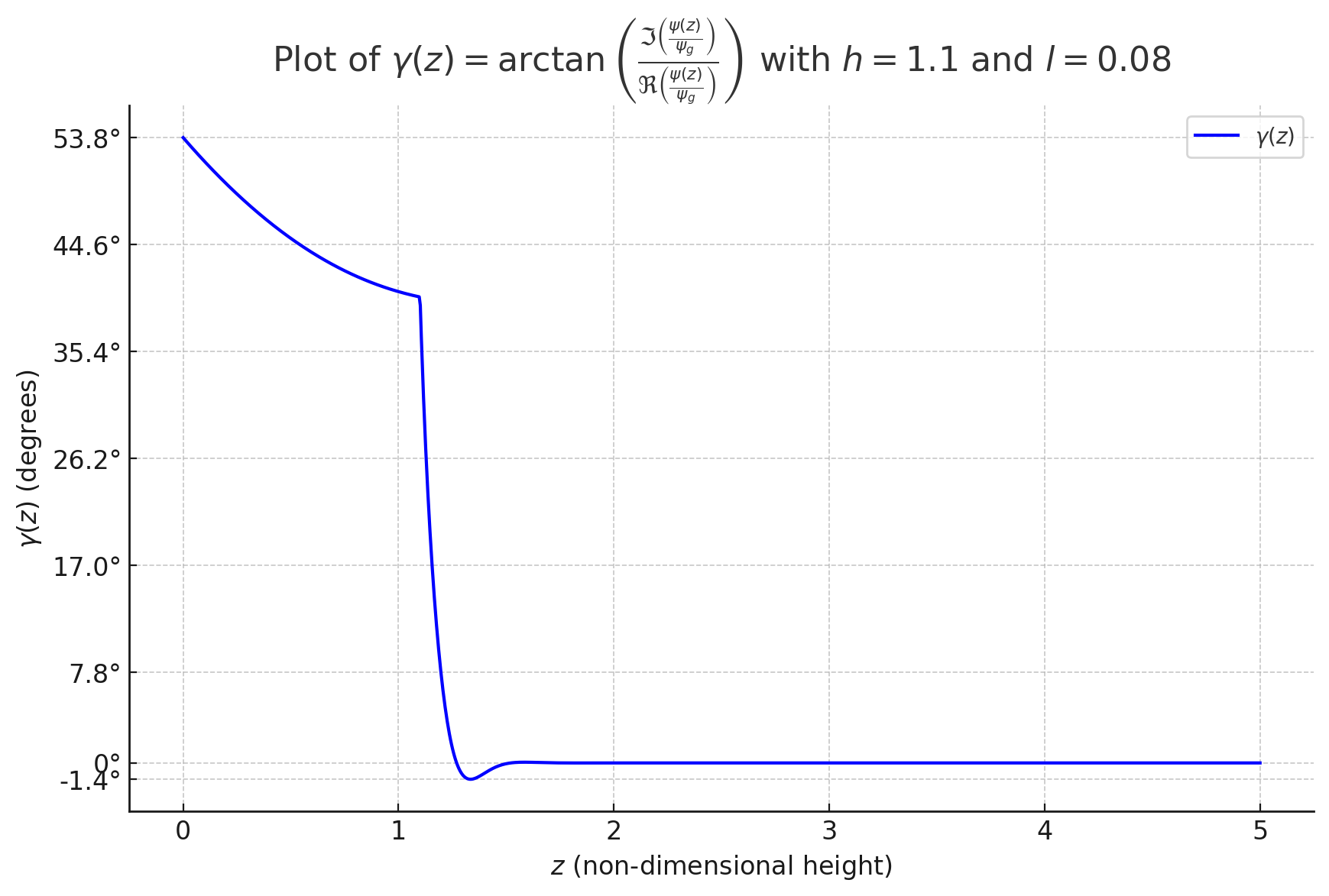}
        \caption{Small eddy viscosity ($l=0.08$) yield a deflection angle over $45^{\circ}$.}
        \label{fig:sub1}
    \end{subfigure}%
    \hspace{0.05\textwidth} 
    \begin{subfigure}[b]{0.47\textwidth} 
        \centering
        \includegraphics[width=\linewidth]{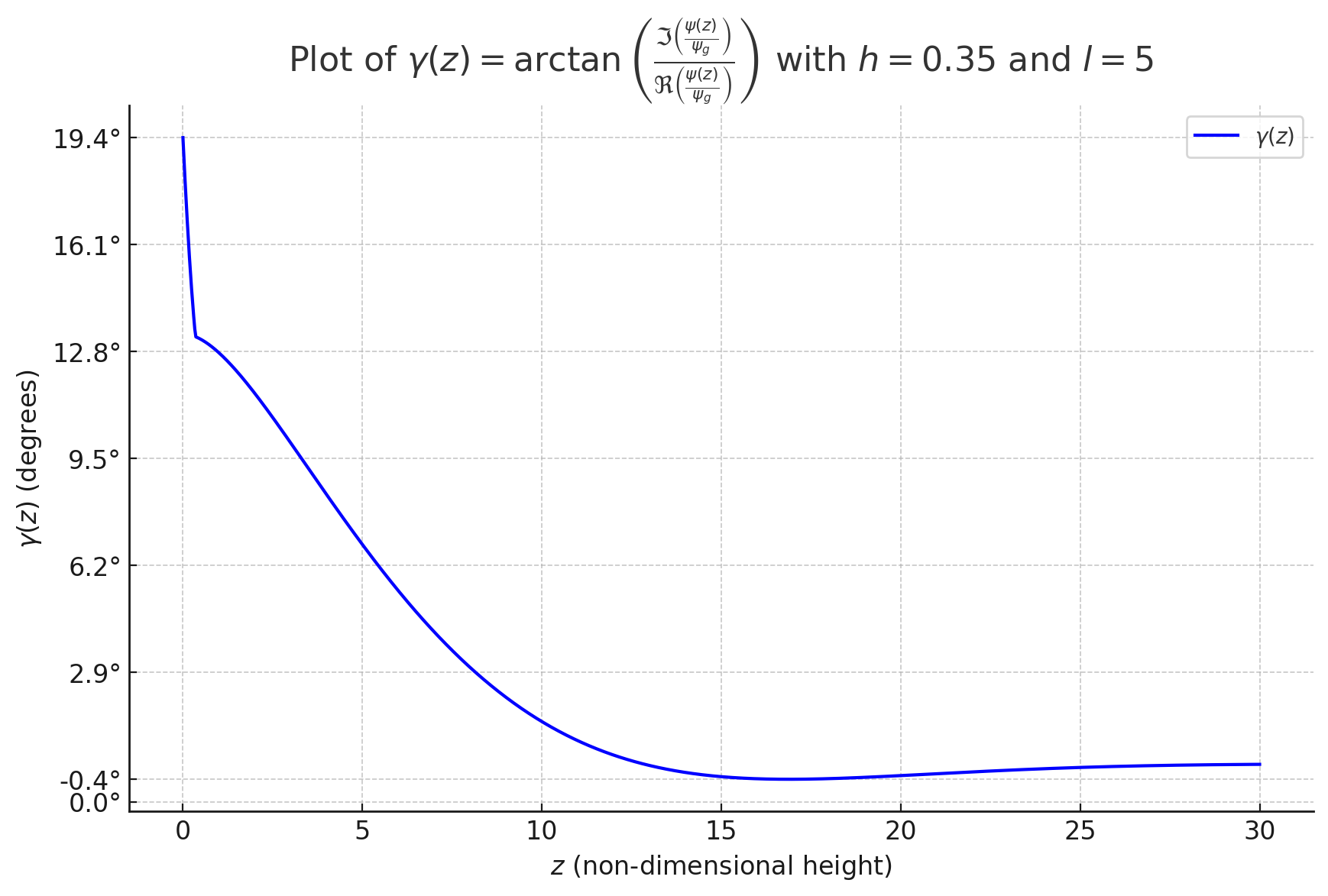}
        \caption{Large eddy viscosity ($l=5$) yield a deflection angle under $45^{\circ}$.}
        \label{fig:sub2}
    \end{subfigure}
    \caption{Comparison of deflection angles for different values of $l$ and $h$.}
    \label{fig:comparison}
\end{figure}

\begin{figure}[ht!]
    \centering
    \begin{subfigure}[b]{0.47\textwidth}
        \centering
        \includegraphics[width=\linewidth]{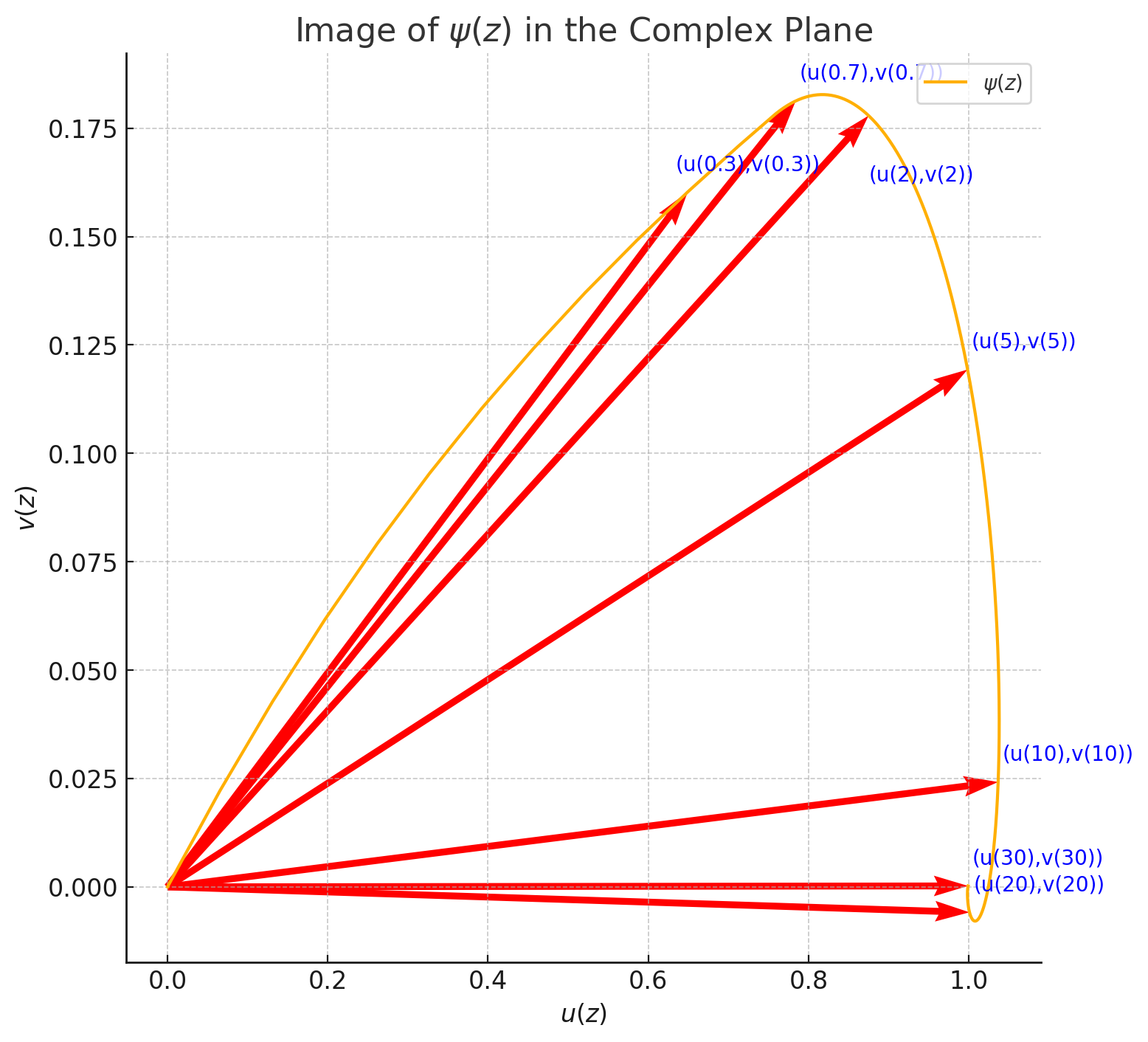}
        \caption{Small eddy-viscosity, high jump-point: \\$l=0.08$, $h=1.1$.}
        \label{fig:sub11}
    \end{subfigure}%
    \hspace{0.05\textwidth} 
    \begin{subfigure}[b]{0.47\textwidth}
        \centering
        \includegraphics[width=\linewidth]{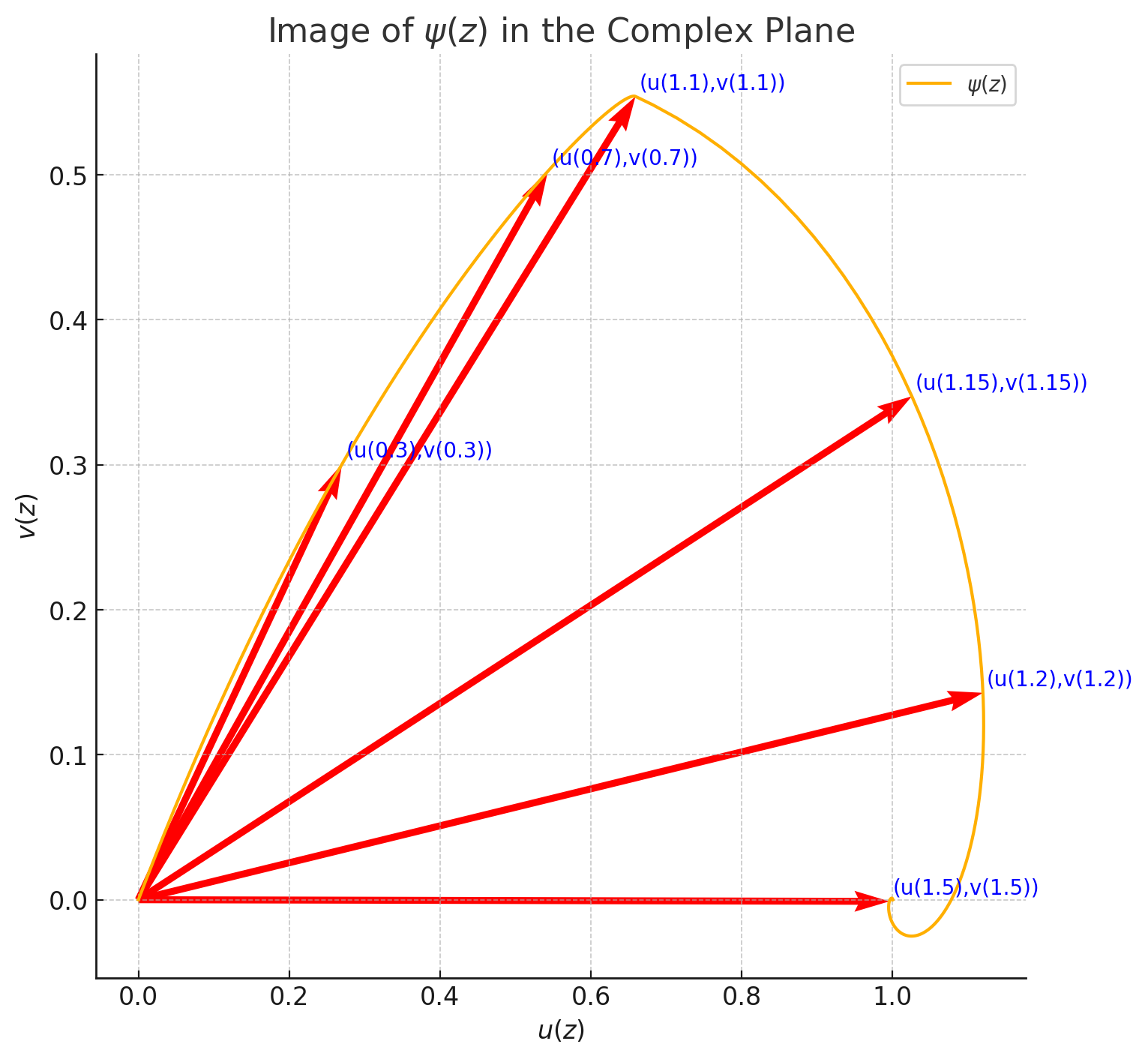}
        \caption{Large eddy viscosity, low jump-point:\\ $l=5$, $h=0.35$.}
        \label{fig:sub22}
    \end{subfigure}
    \caption{Velocity hodographs for different values of $l$ and $h$.}
    \label{fig:comparison}
\end{figure}
Another way to represent the Ekman spiral is via a velocity hodographs, see Figures \ref{fig:sub11},\ref{fig:sub22}. In particular, they can be visualized as looking the at the spiral from the top with according slope generated by the velocity vectors.
They have the disadvantage, that, due to to vanishing velocity at the bottom of the Ekman layer, it is rather hard to observe the exact deflection angle. Nevertheless they give an insight on the behaviour of the spiral throughout the Ekman layer, including convergence rate, jump-points and rotation.

\subsection{Analytic properties}\label{anaprop}
When observing the examples above, the next natural question is: What happens to the deflection, angle when we take the extreme cases?\\
We rewrite \eqref{312}:
\begin{equation}\label{defl}
    \tan(\gamma_0)=\frac{(1+l)^2e^{2h}-(1-l)^2e^{-2h}+2(1-l^2)\sin(2h)}{(1+l)^2e^{2h}-(1-l)^2e^{-2h}-2(1-l^2)\sin(2h)}.
\end{equation}
Let $l\in\mathbb{R}^+$ be fixed and $h\rightarrow0$ or $h\rightarrow\infty$, then in both cases clearly $\gamma_0=\pi/4=45^{\circ}$. In both cases, we have constant eddy viscosity, hence it makes sense, that we get the same result as in Ekman's paper \cite{Ekman}.
Now regarding small and large upper layer eddy viscosity,  where jump points are adapted in succession, we obtain the following results:\\
Let first $l\rightarrow0$ and then $h\rightarrow0$. Then $\gamma_0=\pi/2=90^{\circ}$.
Indeed, we directly have
    $$\lim\limits_{l \rightarrow 0}\tan(\gamma_0)=\frac{\sinh(2h)+\sin(2h)}{\sinh(2h)-\sin(2h)}.$$
    Letting $h\rightarrow0$ and using L'Hospital's rule we have $\tan(\gamma_0)\rightarrow \infty$, which corresponds to $\gamma_0=\pi/2=90^{\circ}$.\\
Let now $l\rightarrow\infty$ and then $h\rightarrow0$. Then $\gamma_0=0=0^{\circ}$,
since after using L'Hospital's rule twice to Equation \eqref{defl}, we have
    $$\lim\limits_{l \rightarrow 0}\tan(\gamma_0)=\frac{\sinh(2h)-\sin(2h)}{\sinh(2h)+\sin(2h)}.$$ 
    Letting now $h\rightarrow0$ and using L'Hospital's rule again, we have $\tan(\gamma_0)\rightarrow 0$, which corresponds to $\gamma_0=0=0^{\circ}$.\\
    If we let $h\rightarrow\infty$, instead of $h\rightarrow0$, we have $\gamma_0=\pi/4=45^{\circ}$, by the very same argument as before. \\
In every case $\gamma(z)\in(0,\frac{\pi}{2})$ which means by Equation \eqref{311}, that $\psi_g$ is to the right of $\psi_0$. By the boundary condition \eqref{6}, $\gamma$ turns clockwise with increasing height until it eventually becomes (almost) zero.
\section{Conclusion}
Motivated by Ekman's work in 1905 \cite{Ekman} and the recent papers \cite{Consti1} and \cite{Consti2}, we considered well-known equations of motion for mesoscale steady flow in the Ekman layer in non-equatorial regions of the northern hemisphere. Then, we assumed the eddy viscosity to be a step function with one jump. In addition to well-known boundary conditions we found new conditions to be able to exactly solve the equations of motion with non-uniform eddy viscosity. In particular we wrote down a method to solve these equations for the eddy viscosity being any step function. Then, we showcased examples, where angles over and under $45^{\circ}$ can indeed occur. Next we extremized the problem by letting $h$ and $l$ go to zero and infinity. Especially for extreme $l$ and small $h$ the angle ranges from $0^{\circ}-90^{\circ}$. Even if the highest and lowest angles are very likely not going to occur in the real world, we get an insight on some influences on the angle nevertheless.
\\

\textbf{Acknowledgments}
All data for this paper are properly cited and referred to, and can be found in the references [3-8, 10-12]. This research was supported by the grant Z 387-N of the Austrian Science Fund (FWF). The author thanks Adrian Constantin for helpful discussions. Furthermore the author sincerely appreciates corrections and improvement made by the anonymous referees.

5
\begin{thebibliography}{99}

\bibitem{arya}
Arya, S.P.: Introduction to Micrometeorology, Academic Press, 2001

\bibitem{Berger1998TheBV}
Berger, B., Grisogono, B.: The baroclinic, variable eddy viscosity
ekman layer, Boundary-Layer Meteorology, vol. 87, pp. 363–380, 1998.

\bibitem{Chai}
Chai, T., Lin, C-L.: Newsom RK, Retrieval of Microscale Flow
Structures from High-Resolution Doppler Lidar Data Using an Adjoint
Model. Journal of Atmospheric Sciences, vol. 61, no. 13, pp. 1500–1520,
2004.


\bibitem{Consti2}
Constantin, A., Johnson R.S.: Atmospheric ekman flows with
variable eddy viscosity, Boundary-layer meteorology, vol. 170, no. 3, pp.
395–414, 2019.



\bibitem{Consti4}
Constantin, A.: Frictional effects in wind-driven ocean currents. Geophysical $\&$ Astrophysical Fluid Dynamics, 115, 1–14., 2020

\bibitem{ekmantype}
Constantin, A., Johnson R.S., Ekman-type solutions for shallow-water flows on a rotating sphere: A new perspective on a classical problem. Physics of Fluids 31, 2019


\bibitem{cushman2011introduction}
Cushman-Roisin, B., Beckers, J.: Introduction to Geophysical
Fluid Dynamics: Physical and Numerical Aspects, International
Geophysics. Elsevier Science, 2011.

\bibitem{Consti1}
Dritschel, D.G., Paldor, N., Constantin, A.: The Ekman spiral for piecewise-
uniform viscosity, Ocean Science 16.5, pp. 1089–1093, 2020.


\bibitem{Ekman}
Ekman, V.W.: On the influence of the earth’s rotation on ocean-currents,
in Ark. Mat. Astron. Fys., pp. 1–52, 1905.

\bibitem{Estoque}
Estoque, M.: A numerical model of the atmospheric boundary layer,
Journal of Geophysical Research (1896-1977), vol. 68, no. 4, pp.
17
1103–1113, 1963.

\bibitem{Grisogono1995AGE}
Grisogono, B.: A generalized ekman layer profile with gradually varying
eddy diffusities, Quarterly Journal of the Royal Meteorological Society,
vol. 121, pp. 445–453, 1995.

\bibitem{Grisogono23}
Grisogono, B.: The angle of the near-surface wind-turning in weakly
stable boundary layers, Quarterly Journal of the Royal Meteorological
Society, vol. 137, no. 656, pp. 700–708, 2011

\bibitem{holton2013introduction}
Holton, J.R., Hakim, G.J.: An Introduction to Dynamic Meteorology, Germany, Elsevier Science, 2013.

\bibitem{Madsen1977ARM}
Madsen, O.S.: A realistic model of the wind-induced ekman boundary
layer, Journal of Physical Oceanography, vol. 7, pp. 248–255, 1977.

\bibitem{Marlatt}
Waggy, S., Marlatt, S., Biringen, S.: “Direct simulation of the turbulent
ekman layer: Evaluation of closure models,” Journal of the Atmospheric
Sciences, vol. 69, 2009.

\bibitem{marshall2008atmosphere}
Marshall, J., Plumb, R.A.: Atmosphere, Ocean and Climate Dynamics: An Introductory Text. Netherlands, Elsevier Science, 2007.

\bibitem{2005BoLMe.115..399P}
Parmhed, O., Kos, I.: Grisogono B, An Improved Ekman Layer
Approximation for Smooth eddy Diffusivity Profiles, Boundary-Layer
Meteorology, vol. 115, no. 3, pp. 399–407, 2005.

\bibitem{pickard1990descriptive}
Pickard, G.L., Talley, L.D., Emery, W.J.: Descriptive Physical Oceanography. Elsevier Science, 1990.

\bibitem{Luigi}
Roberti, L.: The Ekman spiral for piecewise-constant eddy viscosity. Applicable Analysis, 2021.

\bibitem{Tan}
Tan, Z-M.: An Approximate Analytical Solution For The Baroclinic And
Variable eddy Diffusivity Semi-Geostrophic Ekman Boundary Layer,
Boundary-Layer Meteorology, vol. 98, no. 3, pp. 361–385, 2001.

\bibitem{seinfeld}
Seinfeld, J.H.M., Pandis S.N.: Atmospheric Chemistry and Physics: From Air Pollution to Climate Change, John Wiley $\&$ Sons, 2016 

\bibitem{stull}
Stull, R. B.: An Introduction to Boundary Layer Meteorology. Kluwer Academic, 1988.



\end{thebibliography}
\end{document}